\documentclass[a4paper, amsfonts, amssymb, amsmath, reprint, showkeys, nofootinbib, twoside,notitlepage]{revtex4-1}
\def\prd{{Phys.~Rev.~D}}        % Physical Review D
\def\prl{{Phys.~Rev.~Lett.}}    % Physical Review Letters
\usepackage{amsmath,amstext}
\usepackage[T1]{fontenc}
\usepackage{amssymb}
\input{epsf}
\usepackage{graphicx}
\usepackage{ae,aecompl}

\usepackage{hyperref}
\usepackage{amsmath}
\usepackage{amssymb}
\usepackage{mathtools}
\usepackage{bm}
\usepackage{cleveref}
\usepackage{tensor}
\usepackage{braket}
\usepackage{enumitem}
\usepackage{mhchem}
\usepackage{cancel}
\usepackage{amsthm}
\usepackage{upgreek}
\usepackage{mathrsfs}

\usepackage{color}

\def\be{\begin{equation}}
\def\ee{\end{equation}}
\def\beq{\begin{eqnarray}}
\def\eeq{\end{eqnarray}}

\theoremstyle{definition}

\theoremstyle{theorem}
\newtheorem{theorem}{Theorem}

\theoremstyle{corollary}

\begin{document}
\title{Bounds on transport from hydrodynamic stability}
\author{L.~Gavassino}
\affiliation{Department of Mathematics, Vanderbilt University, Nashville, TN, USA}

\begin{abstract}
It was recently shown that the dispersion relations describing singularities of retarded two-point functions in causal quantum field theories always satisfy the fundamental inequality $\mathfrak{Im} \, \omega \leq |\mathfrak{Im} \, k|$, and that several rigorous bounds on transport coefficients follow directly from such inequality. Here, we prove that the same inequality, $\mathfrak{Im} \, \omega \leq |\mathfrak{Im} \, k|$, is a necessary condition for a fluid theory to be covariantly stable (i.e. stable in all frames of reference). Hence, the same bounds on transport that follow from causality in quantum field theory also emerge as stability conditions within relativistic hydrodynamics. This intimate connection between bounds from causality and bounds from stability stems from the fact that covariant stability is possible only in the presence of causality. As a quick application, we show that fluids with luminal speed of sound have vanishing viscosities.
\end{abstract}

\maketitle

\section{Introduction} 

The linear response properties of a relativistic fluid can be analysed within quantum field theory through the study of the two-point (thermal) retarded correlator 
\begin{equation}\label{correlator}
G^R(a,b)=-i\Theta(a^0-b^0)\braket{[\psi(a),\psi(b)]}_T ,
\end{equation}
where $a$ and $b$ are two events in Minkowski spacetime, and $\psi$ is a local operator. The singularities of the Fourier transform of $G^R(0,b)$ define a collection of dispersion relations $\omega(k)$, which can be interpreted as the collective excitations of the fluid \cite{FlorkowskiReview2018}. Recently, Heller, Serantes, Spali\'{n}ski, and Withers (HSSW) \cite{Heller2022} have considered the implications that the principle of causality has on these dispersion relations. Starting from the well-known fact that $[\psi(0),\psi(b)]=0$ whenever $b$ is outside the lightcone \cite{Peskin_book}, they were able to derive the following inequality:
\begin{equation}\label{Fundam}
\mathfrak{Im}\, \omega(k) \leq |\mathfrak{Im} \, k|.
\end{equation}  
Then, they used it to prove a number of interesting facts:
\begin{itemize}
\item[(i)] If $\omega(k):\mathbb{C}\rightarrow \mathbb{C}$ is an entire function, then it is a polynomial of at most degree one;
\item[(ii)] The function $\omega(k):\mathbb{C}\rightarrow \mathbb{C}$ does not present poles or essential singularities; 
\item[(iii)] If $\omega(k)$ is a sound mode, and it is analytic for small $k \in \mathbb{C}$, the infrared speed of sound, $c_s=\omega'(0)$, does not exceed the speed of light;
\item[(iv)] For sound modes and diffusive modes, the diffusion coefficient, $D=i\omega''(0)/2$, is non-negative, and it is bounded above by $16/(3\pi K)$, where $K$ is the radius of convergence of the McLaurin series of $\omega(k)$.
\end{itemize}
It is natural to wonder whether these bounds are specific features of dispersion relations arising from an underlying causal quantum field theory, or if they are universal properties of all ``well-behaved'' fluid theories. More precisely: If we pick a generic linear system of partial differential equations describing small perturbations in a fluid, what are the mathematical assumptions that give rise to \eqref{Fundam}?

We work with Minkowski metric $(-,+,+,+)$, and adopt natural units: $c=\hbar=1$. The equilibrium state is assumed uniform (hence, non-rotating).

\section{Proof from stability}

The original derivation of \eqref{Fundam} by Heller, Serantes, Spali\'{n}ski, and Withers \cite{Heller2022}  is outlined in Appendix \ref{AAA}. Here, we present an alternative derivation which, we believe, will clarify the geometric meaning (and the importance) of this fundamental inequality.
We call ``$\Psi$'' the collection of fields of a linearised relativistic hydrodynamic theory. Aligning the $x-$axis with the (complex) wave-vector $p^\mu$ of a perturbation, so that $p^\mu=(\omega,k,0,0)$, we have the following result:
\begin{theorem}\label{theo1}
Let $\Psi(t,x) = \Psi_0 e^{i(kx-\omega t)}$ be a solution of the linearised equations of a relativistic hydrodynamic theory, where $\Psi_0 \neq 0$, $k \in \mathbb{C}$, and $\omega \in \mathbb{C}$ are constant. If the theory is linearly stable in all reference frames, then
\begin{equation}
\mathfrak{Im}\, \omega \leq |\mathfrak{Im} \, k| .
\end{equation}
\end{theorem}
\begin{proof}
Suppose, by contradiction, that $\mathfrak{Im}\, \omega > |\mathfrak{Im} \, k|$. Then, in particular, $\mathfrak{Im}\, \omega >0$. It follows that the real number $v:=\mathfrak{Im} \, k/\mathfrak{Im} \, \omega$ exists and is finite. Furthermore, we have that $|v|=|\mathfrak{Im} \, k|/\mathfrak{Im}\, \omega <1$. Hence, we can perform a Lorentz boost with velocity $v$:
\begin{equation}\label{Boost}
\begin{split}
t={}& \gamma (t'+vx') \, ,\\
x={}& \gamma (x'+vt') \, ,\\
\end{split}
\end{equation}
where $\gamma=(1-v^2)^{-1/2}$. In these new coordinates, the solution now reads $\Psi'(t',x') =\Psi'_0 e^{i(k'x'-\omega' t')}$, with 
\begin{equation}
\begin{split}
\omega'={}& \gamma (\omega-vk) \, ,\\
k'={}& \gamma (k-v\omega) \, .\\
\end{split}
\end{equation}
Taking the imaginary part of this equation, and recalling that $\mathfrak{Im} \, k = v \, \mathfrak{Im}\, \omega$, we finally obtain
\begin{equation}
\begin{split}
\mathfrak{Im} \, \omega'={}& \mathfrak{Im} \, \omega/\gamma >0 \, ,\\
\mathfrak{Im} \, k'={}& 0 \, .\\
\end{split}
\end{equation}
The fact that $k'$ is real implies that $\Psi'(t',x')$ is a Fourier mode, in this new reference frame. The fact that $\mathfrak{Im} \, \omega'$ is positive implies that $|\Psi'| \propto e^{(\mathfrak{Im} \, \omega')t'}$ grows exponentially in time. But the existence of a growing Fourier mode implies that, in the reference frame $\{t',x'\}$, the theory is unstable \cite{Hiscock_Insatibility_first_order}, contradicting our assumption.
\end{proof}
Theorem 1 shows that equation \eqref{Fundam}, and thus all the bounds that follow from it, is a necessary condition for covariant stability. Indeed, in their original derivation based on correlators, HSSW \cite{Heller2022} needed to assume that $G^R$ is a tempered distribution, which is similar to requiring stability (see Appendix \ref{AAA}). To appreciate why causality alone is not enough, consider the following equation\footnote{Equation \eqref{counter} is the Bludman-Ruderman model for sound in dense matter \cite{Bludman1968,FoxKuper1970}, with $\mu=1/2$ and front velocity $v_\infty=\sqrt{D}=1$.}:
\begin{equation}\label{counter}
(\partial_t^2-\partial^2_x)^2\Psi +(\partial^2_t - c_s^2 \partial^2_x)\Psi=0.
\end{equation}
Its principal part is the square of the ordinary wave operator, $\partial^2_t-\partial^2_x$. Hence, causality is respected for all values of $c_s$ \cite{CourantHilbert2_book,Rauch_book}. On the other hand, for small gradients, equation \eqref{counter} becomes $(\partial^2_t - c_s^2 \partial^2_x)\Psi \approx 0$, meaning that the system presents sound modes, with infrared speed of sound $c_s$. Hence, the bound (iii) applies, and we must require that $c_s \leq 1$. But this is not a causality condition, since \eqref{counter} is always causal. It is a stability condition. In fact, equation \eqref{counter} presents 4 dispersion relations:
\begin{equation}\label{omega2}
2\omega^2 = 1+2k^2 \pm \sqrt{1+4k^2(1-c_s^2)} \, .
\end{equation}
If $c_s^2 \leq 1$, then $\omega^2>0$ for real $k$, meaning that $\mathfrak{Im} \, \omega=0$ on all Fourier modes. If, instead, $c_s^2>1$, the dispersion relations \eqref{omega2} become complex for large $k \in \mathbb{R}$, and there is one growing Fourier mode, signalling an instability.

\section{Proof from causality}

Although causality in itself is not enough to guarantee \eqref{Fundam}, still, it is a necessary condition for covariant stability. The reason is that, if causality is violated, different observers disagree on whether a dissipative system is stable or not \cite{GavassinoSuperluminal2021}. If, instead, causality holds, then all observers must agree on the stability properties of the system. In particular, if there is at least one reference frame in which a causal theory is stable, then such theory will automatically be stable in all reference frames, and \eqref{Fundam} will also hold. For this reason, we have the following
\begin{theorem}\label{Theocaus}
Let $\Psi(t,x) = \Psi_0 e^{i(kx-\omega t)}$ be a solution of the linearised equations of a relativistic hydrodynamic theory, where $\Psi_0 \neq 0$, $k \in \mathbb{C}$, and $\omega \in \mathbb{C}$ are constant. If the theory is linearly stable in the reference frame $\{t,x \}$, and it is causal, then
\begin{equation}
\mathfrak{Im}\, \omega \leq |\mathfrak{Im} \, k| .
\end{equation}
\end{theorem}
\begin{proof}
If $\mathfrak{Im} \, k=0$, then $\Psi(t,x)$ is a Fourier mode, and stability directly demands that $\mathfrak{Im}\, \omega \leq 0= |\mathfrak{Im} \, k|$. Thus, let us focus on the case with $\mathfrak{Im} \, k \neq 0$. For convenience, we orient the $x$-axis (by means of an $xy$-rotation) in such a way that $\mathfrak{Im} \, k > 0$, and we construct a new solution, $\Psi^\star(t,x)$, with the following initial data (see figure \ref{fig:PsiStar}):
\begin{figure}
\begin{center}
\includegraphics[width=0.45\textwidth]{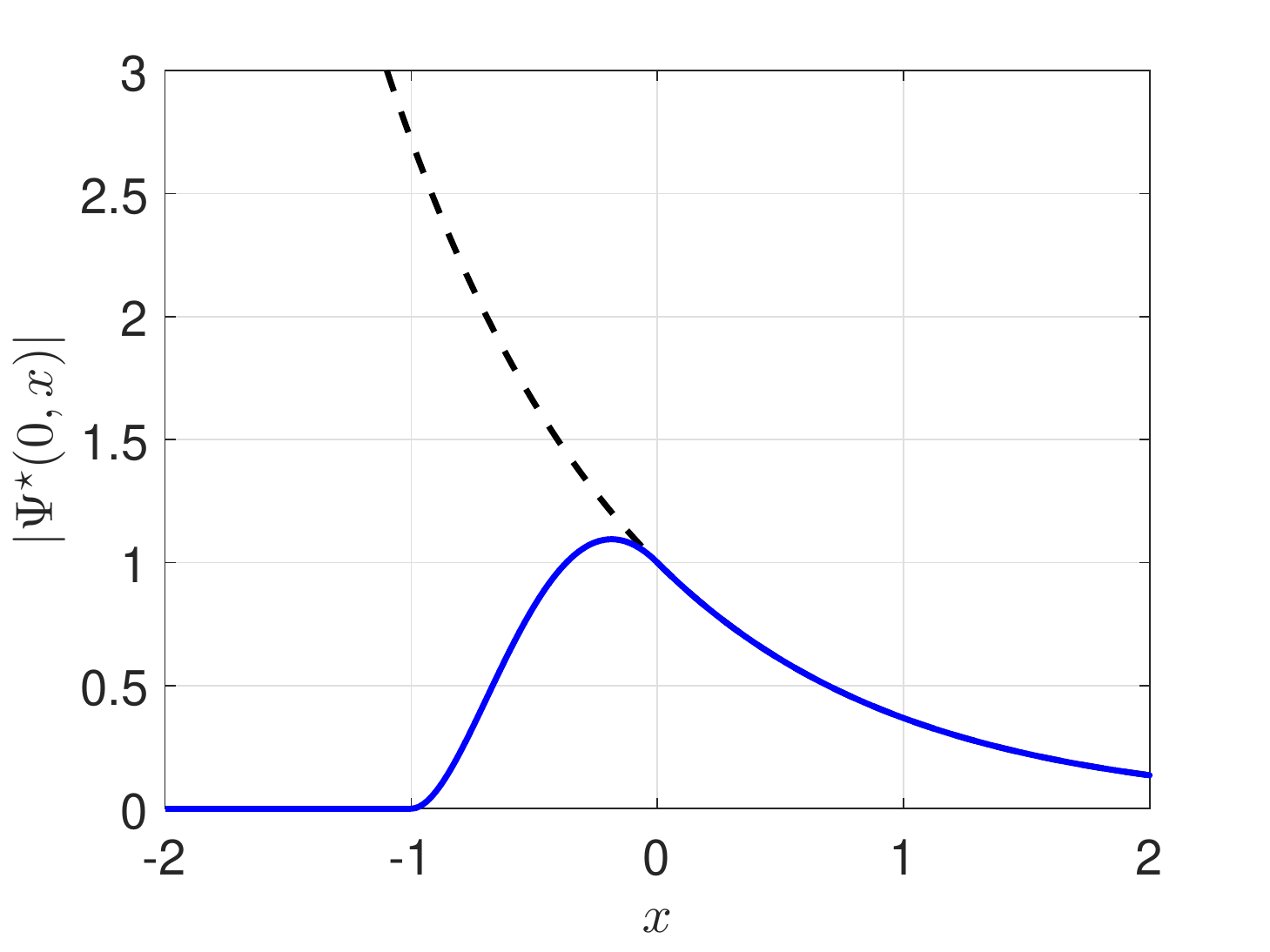}
\includegraphics[width=0.43\textwidth]{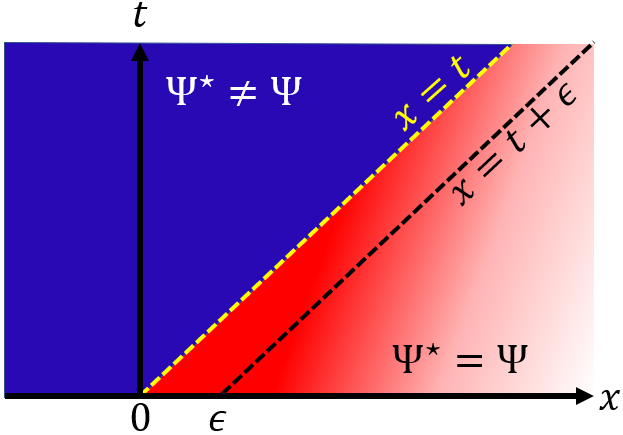}
	\caption{Upper panel: Initial data for $|\Psi^\star|$ (blue line) compared to the initial data for $|\Psi|$ (dashed line), in natural units. Both functions are smooth, and they agree for positive $x$, but $\Psi^\star$ has been constructed to vanish for $x\leq -1$. Lower panel: Minkowski diagram of $\Psi^\star(t,x)$. In the blue region, $\Psi^\star$ may differ from $\Psi$, and its exact value depends on the details of the fluid equations. In the red region, $\Psi^\star$ agrees with $\Psi$ (by causality), and the shades of red are a qualitative sketch of the intensity of $|\Psi^\star|$ (red$\,=\,$large, white$\,=\,$small). By stability, we know that $\Psi^\star$ cannot diverge along the null paths $x{=}t{+}\epsilon$ (with $\epsilon \geq 0$), and this produces the inequality $\mathfrak{Im}\, \omega \leq |\mathfrak{Im} \, k| $.}
	\label{fig:PsiStar}
	\end{center}
\end{figure}
\begin{equation}\label{thetuzzo}
\Psi^\star(0,x)=\Theta^\star(x)\Psi(0,x) .
\end{equation}
Here, $\Theta^\star \geq 0$ is an arbitrary smooth function that equals $1$ for $x \geq 0$ and equals $0$ for $x \leq -1$.
The initial profile of $\Psi^\star$ is smooth, and it has a well defined Fourier transform. In particular, $|\Psi^\star(0,x)|\propto \Theta^\star(x)e^{- (\mathfrak{Im} \, k) x}$, with $\mathfrak{Im} \, k > 0$. Stability implies that $\Psi^\star(t,x)$ cannot diverge for large $t$, because $\Psi^\star(0,x)$ is a smooth perturbation of finite square integral norm \cite{Hiscock_Insatibility_first_order}. On the other hand, since the initial data of $\Psi^\star$ agrees with the initial data of $\Psi$ for $x \geq 0$, and the theory is causal, we necessarily have that $\Psi^\star=\Psi$ for $x\geq t \geq 0$ \cite{GavassinoSuperluminal2021,Susskind1969,Hawking1973,Wald} (see figure \ref{fig:PsiStar}, lower panel). In particular, for any $\epsilon \geq 0$ and $t \geq 0$, we have that
\begin{equation}
|\Psi^\star(t,t{+}\epsilon)|=|\Psi(t,t{+}\epsilon)|\propto e^{(\mathfrak{Im}\, \omega-\mathfrak{Im} \, k)t}.
\end{equation}
The only way to avoid a divergence at large $t$ is to require that $\mathfrak{Im}\, \omega \leq \mathfrak{Im} \, k$, which is what we wanted to prove.
\end{proof}

Let us see a couple of examples.

\textit{Example 1 -} Consider the Israel-Stewart theory for bulk viscosity \cite{Israel_Stewart_1979,Hishcock1983,Salmonson1991,
BulkGavassino}. In the linear regime, the characteristic velocity is given by \cite{Causality_bulk}
\begin{equation}
c_{\text{ch}}^2 = c_s^2 + \dfrac{\zeta}{(\rho+P)\tau}  ,
\end{equation}
where $c_s$ is the infrared speed of sound, $\zeta$ is the bulk viscosity coefficient, $\tau$ is the bulk relaxation time, and $\rho+P$ is the enthalpy density. Stability in the rest frame implies that $c_s^2,(\rho+P),(\tau/\zeta)>0$ \cite{Hishcock1983,GavassinoGibbs2021}. Hence, $0 < c_s^2 \leq c_{\text{ch}}^2$. But causality demands $c_{\text{ch}}^2 \leq 1$, so that we necessarily have $0 < c_s^2 \leq 1$, and bound (iii) of the introduction is indeed recovered.

\textit{Example 2 -} Consider the Cattaneo model \cite{cattaneo1958,rezzolla_book,
GavassinoCausality2021}:
\begin{equation}\label{Cattaneuz}
\dfrac{\partial^2_t \Psi}{w^2}+\dfrac{\partial_t \Psi}{D} = \partial^2_x \Psi \, .
\end{equation}
Stability in the rest frame implies that $w^2>0$ and $D>0$, and causality requires $w^2 \leq 1$ ($w$ is the characteristic speed). The dispersion relations associated to \eqref{Cattaneuz} are 
\begin{equation}\label{cattuzdisp}
\omega_\pm(k) =- \dfrac{iw^2}{2D} \bigg[ 1 \pm  \sqrt{1-4k^2 \dfrac{D^2}{w^2} } \bigg] .
\end{equation}
The dispersion relation $\omega_-(k)$ is a diffusive mode, with diffusion coefficient $D$. Its radius of convergence around $k=0$ is $K=w/(2D)$. Hence, rest frame stability and causality imply that $0 < D \leq 1/(2K)$, in agreement with the bound (iv), and with the discussion of \cite{Heller2022}.

\section{Three quick applications}

We conclude the article by discussing some further consequences of equation \eqref{Fundam}, and of our theorems.

\subsection{Superluminal sound waves are unstable}

The authors of \cite{Heller2022} showed that bound (iii), i.e. $c_s \leq 1$, follows directly from \eqref{Fundam}. In Theorem \ref{theo1}, we have shown that \eqref{Fundam} is a necessary condition for covariant stability. It follows that any system with $c_s>1$ is necessarily unstable, independently from whether it is dissipative or not, and from whether it is causal or not. In fact, suppose that, for small $k \in \mathbb{C}$, we can express $\omega(k)$ as follows:
\begin{equation}\label{quq}
\omega(k) \approx c_s k  \, .
\end{equation}
A dispersion relation of the form \eqref{quq}, with $k$ \textit{complex}, exists for any partial differential equation that reduces to $(\partial^2_t -c_s^2 \partial^2_x)\Psi \approx 0$ for small gradients.
Then, setting $k=i \epsilon$ in \eqref{quq}, with $\epsilon >0$, we obtain solutions of the form $\Psi(t,x)=\Psi_0 \, e^{-\epsilon (x-c_s t)}$. If $c_s>1$, we can perform the boost \eqref{Boost}, with velocity $v=c_s^{-1}$. In the new reference frame, we have
$\Psi'(t',x')= \Psi_0' \, e^{\epsilon \, c_s t'/\gamma}$, which is a growing Fourier mode, signalling an instability.

%This simple argument settles an old debate \cite{Ellis2007}. 
It is often said that, in relativistic fluids, the adiabatic speed of sound $c_s=\partial P/\partial \rho$ (at constant specific entropy) must be subluminal, to ensure causality. This bound was criticised in \cite{FoxKuper1970,Caporaso1979}, because in dispersive media the signal propagation is not determined by $c_s$. However, it was repeatedly noticed \cite{Bludman1970,Krotscheck1978,Adams2006,Camelio2022} that any attempt to construct a system with $c_s>1$ that is both causal and stable is doomed to fail. Now we have provided an intuitive and general proof that $c_s \leq 1$ is a necessary condition for \textit{covariant stability}, rather than for causality.

\subsection{Luminal sound waves cannot be dissipated}

The above analysis cannot be used to rule out the limiting case $c_s =1$. Indeed, \citet{Zeldovich1961} has shown that some fluids can get very close to this limit. We can prove, however, that a fluid whose speed of sound is exactly $1$ cannot be viscous. Consider a sound-type dispersion relation which, for small $k \in \mathbb{C}$, can be expanded as
\begin{equation}\label{ksquared}
\omega(k)=k-iDk^2 +\mathcal{O}(k^3) \, ,
\end{equation}
where $D \in \mathbb{R}$. Stability in the rest frame requires $D \geq 0$. Furthermore, if we set $k=i \epsilon$ (with $\epsilon > 0$), equation \eqref{ksquared} becomes $\omega= i(\epsilon+D\epsilon^2)+\mathcal{O}(\epsilon^3)$. In the limit of small $\epsilon$, the constraint \eqref{Fundam} is valid only if $D$ vanishes identically. But for a simple fluid, $D$ is related to the bulk and shear viscosity coefficients, $\zeta$ and $\eta$, through the relation \cite{Kovtun2019} 
\begin{equation}
D = \dfrac{\zeta +4\eta/3}{2(\rho+P)} \, .
\end{equation}
Since both $\zeta$ and $\eta$ are non-negative (by stability), the only way to have $D=0$ is by setting $\zeta=\eta=0$.

\subsection{A new understanding of covariant stability}

Theorem \ref{theo1} can be reformulated in a manifestly covariant form by saying that, if a theory is covariantly stable, then all complex exponential solutions of the form $\Psi(x^\mu) {=} \Psi_0 e^{ip_\mu x^\mu}$ (where $\Psi_0 {\neq} 0$ and $p_\mu {\in} \mathbb{C}^4$ are constant) must be such that $\mathfrak{Im}\, p^\mu$ is \textit{not} timelike future directed. This is a novel and intuitive way of understanding  the very notion of ``covariant stability'' from a geometric perspective. In fact, suppose that $\mathfrak{Im}\, p^\mu$ is timelike future directed. Then, we have that $\mathfrak{Im}\, p^\mu= \Gamma u^\mu$, where $u^\mu$ is a four-velocity, and $\Gamma>0$. It then follows that
\begin{equation}
|\Psi(x^\mu)| \propto e^{-(\mathfrak{Im}\, p_\mu)x^\mu}= e^{-\Gamma u_\mu x^\mu}  \, .
\end{equation} 
Hence, $\Psi$ is a growing Fourier mode in the reference frame of an observer in motion with four-velocity $u^\mu$.

Let's consider a simple example. The diffusion equation $\partial_t \Psi=D \partial^2_x \Psi$ admits solutions with $p^\mu{=}(iD\epsilon^2,i\epsilon)$, for arbitrary $\epsilon \in \mathbb{R}$. When $D\epsilon \,{>} \,1$, the vector $\mathfrak{Im}\, p^\mu$ is timelike future directed, and an observer who moves with velocity $v=(D\epsilon)^{-1}$ detects a growing Fourier mode, with growth rate $\Gamma= D\epsilon^2/\gamma$. This shows, once again \cite{Kost2000,GavassinoLyapunov_2020,
GavassinoFronntiers2021}, that the boosted diffusion equation is unstable.

\section{Conclusions}

We have proved that any linear system of partial differential equations which describes the dynamics of small fluid perturbations around equilibrium must satisfy \eqref{Fundam} as a necessary (but not sufficient\footnote{A counter-example is the differential equation $\partial^2_t \Psi=0$. It is unstable, because it admits growing solutions of the form $\Psi(t,x)=f(x)t$, with $f \in L^2(\mathbb{R})$. On the other hand, its only dispersion relation is $\omega(k)=0$, which is consistent with \eqref{Fundam}.}) condition for hydrodynamic stability, see figure \ref{fig:dotted}. This provides an alternative way of understanding the domain of applicability of the bounds found by Heller, Serantes, Spalinski, and Withers \cite{Heller2022}. These bounds are, indeed, a universal property of near-equilibrium dynamics in special relativity, and they hold both in a classical and in quantum field theory.

The connection of these bounds with causality is a manifestation of the general mechanism by which covariant stability is possible only for causal systems \cite{GavassinoSuperluminal2021}. In fact, we were able to prove a theorem, according to which, if a system is stable in one reference frame, and causal, then \eqref{Fundam} automatically follows.

The existence of a connection between causality (or stability) and the values of the transport coefficients has been known for decades \cite{Hishcock1983,OlsonLifsh1990,DenicolKodama,Romatschke2010,Pu2010,
BemficaDNDefinitivo2020,
GavassinoStabilityCarter2022}. The real novelty of the bounds (i-iv) is that they are not ``theory-specific'', i.e. they must be respected by any stable relativistic hydrodynamic theory, no matter how complicated.

\begin{figure}[h!]
\begin{center}
\includegraphics[width=0.43\textwidth]{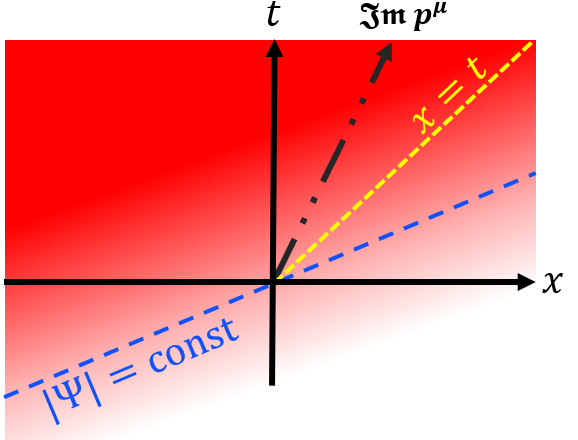}
	\caption{The intensity $|\Psi|\propto e^{(\mathfrak{Im} \,\omega)t-(\mathfrak{Im} \, k)x}$ of a perturbation is constant along the hypersurface $t=(\mathfrak{Im}\, k/\mathfrak{Im}\, \omega) \, x$.
If the inequality \eqref{Fundam} is violated (i.e. if $\mathfrak{Im}\, \omega > |\mathfrak{Im}\, k|$), then such hypersurface is spacelike (blue dashed line), and an observer whose four-velocity is collinear to $\mathfrak{Im} \, p^\mu=(\mathfrak{Im}\, \omega,\mathfrak{Im}\, k)$ detects an unstable Fourier mode \cite{Hiscock_Insatibility_first_order}, namely a sinusoidal perturbation whose amplitude is uniform in space and grows exponentially in time (shades of red).} 
	\label{fig:dotted}
	\end{center}
\end{figure}

\section*{Acknowledgements}

This work was supported by a Vanderbilt's Seeding Success Grant. I would like to thank M. Disconzi and J. Noronha for fruitful discussions. I am also grateful to  M. Heller, A. Serantes, M. Spali\'{n}ski, B. Withers, and P. Romatschke for reading the manuscript and providing useful comments.

\appendix

\section{The HSSW derivation}\label{AAA}

Despite formally the HSSW argument for \eqref{Fundam} relies on the study of the quantum retarded correlator \eqref{correlator}, the main idea is very general, and it can be naturally adapted to relativistic hydrodynamics (and more generally to classical field theory) by replacing $G^R$ with the fundamental solution of the system of partial differential equations. Here, we briefly outline how one can systematically apply the HSSW reasoning in this context.

\subsection{Applying the argument to hydrodynamics}

Consider a single\footnote{It is well-known that, from any system of linear differential equations with
constant coefficients, a single linear differential equation with constant coefficients can be obtained for any of the unknown functions \cite{CourantHilbert2_book}.} partial differential equation with constant coefficients of the form $\mathcal{L}(\partial_\mu)\varphi=0$, where $\varphi(x^\mu)$ is a scalar field, and $\mathcal{L}(\partial_\mu)$ is a linear differential operator. A field configuration of the form $\varphi=e^{ip_\mu x^\mu}$, with constant $p^\mu \in \mathbb{C}^4$, is a solution to the partial differential equation if and only if
\begin{equation}
\mathcal{L}(ip_\mu)=0.
\end{equation}
This defines the (complex) dispersion relations of the system. The fundamental solution, $G(x^\mu)$, of the partial differential equation is the solution to the Cauchy problem
\begin{equation}\label{gaga}
\left\{ \begin{aligned} 
  & \mathcal{L}(\partial_\mu)G = \delta^4(x^\mu)  , \\
  & G(t<0) = 0 .
\end{aligned} \right.
\end{equation}
Now, we make two basic assumptions:
\begin{itemize}
\item[(a)] $G$ exists, and it is a tempered distribution. This is similar to requiring that the fluid model is stable. In fact, for unstable systems, $G$ usually grows exponentially in time, and it is not expected to be a tempered distribution.
\item[(b)] The support of $G$ does not exit the future lightcone. This is automatically true if the system is causal. 
\end{itemize}
Then, the function
\begin{equation}\label{GGG}
\tilde{G}(p^\mu)= \int G(x^\mu)e^{-ip_\mu x^\mu} d^4 x
\end{equation}
exists and is finite whenever $\mathfrak{Im} \, p^\mu$ lays inside the open future lightcone. In fact, $G$ is a tempered distribution, and $\big|e^{-ip_\mu x^\mu}\big|=e^{(\mathfrak{Im} \, p_\mu) x^\mu}$ decays exponentially over the support of the integrand (because $\mathfrak{Im} \, p^\mu$ is future directed timelike, and $x^\mu$ is future directed non-spacelike), so that it plays the role of a Schwartz test function. On the other hand, if we multiply the first equation of \eqref{gaga} by $e^{-ip_\mu x^\mu}$, and we integrate over the spacetime, we obtain
\begin{equation}\label{LLL}
\mathcal{L}(ip_\mu)\tilde{G}=1 .
\end{equation}
Now, suppose that there is a solution to $\mathcal{L}(\partial_\mu)\varphi=0$ of the form $\varphi=e^{ip_\mu x^\mu}$ where $\mathfrak{Im} \, p^\mu$ lays inside the open future lightcone. Then, $\mathcal{L}(ip_\mu)$ vanishes for this choice of $p_\mu$. But this contradicts equation \eqref{LLL}, since we know that $\tilde{G}(p_\mu)$ must be finite (being the action of a tempered distribution on a Schwartz function). We conclude that all the dispersion relations of the theory are such that $\mathfrak{Im} \, p^\mu$ sits outside the open future lighcone. Setting $\mathfrak{Im} \, p^\mu=(\mathfrak{Im} \, \omega,\mathfrak{Im} \, k,0,0)$, we finally obtain equation \eqref{Fundam}.

\subsection{Comment on temperedness and stability}

Despite the derivation of \eqref{Fundam} by HSSW outlined above is remarkably simple and general, it is not immediate to translate assumption (a) into a simple set of mathematical conditions on the differential equation $\mathcal{L}(\partial_\mu)\varphi=0$. As already noted by \cite{Heller2022}, requiring that $G$ is tempered is ``similar'' to requiring stability. However, the two concepts are not equivalent, since there are unstable systems whose Green function is tempered. 

For example, the equation\footnote{An example of a system that obeys \eqref{ehibarista} is the pressureless perfect fluid (the ``dust'' \cite{rezzolla_book}), whose linearised fluid equations (in natural units) are $\partial_t \rho + \partial_x u=0$, $\partial_t u=0$, where $\rho$ is the energy density and $u$ is the flow velocity. Taking the time derivative of the first equation, and invoking the second, we indeed obtain $\partial^2_t \rho=0$. The general solution to the linearised dust dynamics is manifestly unstable: $u(t,x)=f(x)$, $\rho(t,x)=g(x)-f'(x)t$.} 
\begin{equation}\label{ehibarista}
\partial^2_t \varphi=0
\end{equation} 
is unstable, as it admits growing solutions of the form $\varphi(t,x)=f(x) t$ (where $f$ can be taken square integrable), but its retarded Green function is $G(t,x)=\delta(x)\Theta(t) t$, which is tempered, as it grows linearly in $t$. Furthermore, there also exist tempered distributions that grow exponentially in time, e.g. $e^t \cos(e^t)$, which is the distributional derivative of $\sin(e^t)$. Hence, the mathematical connection between hydrodynamic stability and assumption (a) is not straightforward, and it still needs to be established rigorously.

%
%\section{Stability analysis of our toy-model}\label{BBB}
%
%If we apply a Lorentz boost to equation \eqref{counter}, find that the boosted dispersion relations are solutions of the depressed quartic equation $\omega^4-A\omega^2-B\omega+F=0$, with
%\begin{equation}
%\begin{split}
%A ={}& 2 k^2 +\gamma^2(1-v^2 c_s^2) \, ,  \\
%B ={}& 2\gamma^2(1-c_s^2)vk \, , \\
%F ={}& \gamma^2(c_s^2-v^2)k^2+k^4 \, . \\
%\end{split}
%\end{equation}

\bibliography{Biblio}

\label{lastpage}

\end{document}